\documentclass[11pt]{article}

\usepackage{times}
\usepackage{nicefrac}
\usepackage{epsfig}
\usepackage{graphicx}
\usepackage{amssymb}
\usepackage{cite}
\usepackage{amsthm}
\usepackage{amsmath, fullpage}
\usepackage{complexity}
\usepackage[textsize=scriptsize,textwidth=2.1cm]{todonotes}
\usepackage{thmtools}
\usepackage{thm-restate}
\usepackage{multirow}
\usepackage[T1]{fontenc}
\usepackage{lipsum}
\usepackage{multicol}
\usepackage{changepage}

\def \eps {\varepsilon}
\newcommand{\RR}{\mathbb{R}}

\newcommand{\MaxRect}{\textsc{Max-Weight Rectangle} }

\newcommand{\MaxSquare}{\textsc{Max-Weight Square} }
\newcommand{\MaxSubarray}{\textsc{Maximum Subarray} }
\newcommand{\MaxSubarraySquare}{\textsc{Maximum Square Subarray} }
\newcommand{\WDepth}{\textsc{Weighted Depth} }
\newcommand{\APSP}{\textsc{All-Pairs Shortest Paths} }
\newcommand{\kClique}{\textsc{Max-Weight $k$-Clique} }
\newcommand{\CliqueMatching}{\textsc{Max-Weight Clique without Matching} }
\newcommand{\CliqueDetect}{\textsc{$k$-Clique-Detection} }
\newcommand{\NegTriangle}{\textsc{Negative Triangle} }
\newcommand{\PosTriangle}{\textsc{Positive Triangle} }

\newcommand{\wClique}[1]{\textsc{Max-Weight $\left(#1\right)$-Clique} }

\newcommand{\CMaxSum}{\textsc{Central Max-Sum} }
\newcommand{\CMC}{\textsc{Central Maximum Combination} }
\newcommand{\MComb}{\textsc{Maximum Combination} }
\newcommand{\MCube}{\MaxSubarraySquare }

\newcommand{\MfComb}{\textsc{Maximum $4$-Combination} }

\newcommand{\MR}{\MaxSubarray }
\newcommand{\MaxSub}[1]{\textsc{Max-weight $#1$-Subgraph problem} }
\newcommand{\CMRS}{\textsc{Central Maximum Subarray Sum} }
\newcommand{\CMRC}{\textsc{Central Maximum Subarray Combination} }
\newcommand{\MRC}{\textsc{Maximum Subarray Combination} }

\newtheorem{theorem}{Theorem}
\newtheorem{lemma}[theorem]{Lemma}
\newtheorem{definition}{Definition}
\newtheorem*{remark}{Remark}

\begin{document}
\title{Tight Hardness Results for Maximum Weight Rectangles}
		
\author{
	Arturs Backurs\iftrue\footnote{\texttt{backurs@mit.edu}}\fi\\ MIT 
	\and 
	Nishanth Dikkala\iftrue\footnote{\texttt{ndikkala@mit.edu}}\fi\\ MIT 
	\and 
	Christos Tzamos\iftrue\footnote{\texttt{tzamos@mit.edu}}\fi\\ MIT 
}

\date{}

\maketitle

\begin{abstract}
Given $n$ weighted points (positive or negative) in $d$ dimensions, what is the axis-aligned box 
which maximizes the total weight of the points it contains? 

The best known algorithm for this problem is based on a reduction to 
a related problem, the \WDepth problem [T. M. Chan, FOCS'13], and runs in time 
$O(n^d)$. It was conjectured [Barbay et al., CCCG'13] that this runtime is tight up to subpolynomial
factors. We answer this conjecture affirmatively by providing a matching conditional lower bound.
We also provide conditional lower bounds for the special case when points are arranged in a grid
(a well studied problem known as \MaxSubarray problem) as well as for other related problems.

All our lower bounds are based on assumptions that the best known algorithms for the \APSP problem (APSP)
and for the \kClique problem in edge-weighted graphs are essentially optimal.
\end{abstract}

\section{Introduction}
Consider a set of points in the plane.
Each point is assigned a real weight that can be either positive or negative.
The \MaxRect problem asks to find an axis parallel rectangle that maximizes the total weight of the points it contains.
This problem (and its close variants) is one of the most basic problems in computational geometry and is used as a subroutine in many applications \cite{eckstein2002maximum, fukuda1996data, liu2003planar, backer2010mono, agarwal2006hunting}. 
Despite significant work over the past two decades, the best known algorithm runs in time quadratic in the number of points 
\cite{dobkin1996computing, cortes2009bichromatic, barbay2014maximum}.
It has been conjectured that there is no strongly subquadratic time algorithm\footnote{A strongly subquadratic algorithm runs in time $O(N^{2-\eps})$ for constant $\eps>0$.} for this problem \cite{barbay2014maximum}.

An important special case of the \MaxRect problem is when the points are arranged in a square grid. In this case the input is given as an $n \times n$ matrix filled with real numbers and the objective is to compute a subarray that maximizes the sum of its entries \cite{perumalla1995parallel,takaokaguy,smith1987applications,qiu1999parallel,cheng2005improved}.
This problem, known as \MaxSubarray problem, has applications in pattern matching \cite{fischer1993approximations}, data mining and visualization \cite{fukuda1996data} (see \cite{takaokaguy} for additional references). The particular structure of the \MaxSubarray problem allows for algorithms that run in $O(n^3)$, i.e. $O(N^{3/2})$ with respect to the input size $N = n^2$, as opposed to $O(N^2)$ which is the best algorithm for the more general \MaxRect problem.

One interesting question is if this discrepancy between the runtimes of these two very related problems can be avoided. Is it possible to apply ideas from one to improve the runtimes of the other? Despite considerable effort there has been no significant improvement to their runtime other than by subpolynomial factors since they were originally studied.

In this work, we attempt to explain this apparent barrier for faster runtimes by giving evidence of the inherent hardness of the problems. 
In particular, we show that a strongly subquadratic algorithm for \MaxRect would imply a breakthrough for fundamental graph problems.
We show similar consequences for $O(N^{3/2-\eps})$ algorithms for the 
\MaxSubarray problem. Our lower bounds are based on standard hardness assumptions for the \APSP and the \kClique problems and generalize to the higher-dimensional versions of the problems.


\bgroup
\def\arraystretch{1.3}%
\begin{table*}[t] 
\centering
  \begin{tabular}{| c | c | c |}
  \hline
  Problem & {\bf In 2 dimensions} & {\bf In d dimensions} \\ 
  \hline
    \MaxRect & $O(N^2)$ \cite{barbay2014maximum,chan2013klee}& 
    $O(N^d)$ \cite{barbay2014maximum,chan2013klee}\\
  on $N$ weighted points & $\Omega(N^2)$ \textsf{\bf [this work]} & $\Omega(N^d)$ \textsf{\bf [this work]}\\ 
  \hline
  \MaxSubarray
  & $O(n^3)$ \cite{tamaki1998algorithms,takaokaguy} & $O(n^{2d-1})$ \textsf{[Kadane's algorithm]}\\
  on $n\times \cdots \times n$ arrays & $\Omega(n^3)$ \textsf{\bf [this work]} & $\Omega(n^{\nicefrac {3d} {2}})$ \textsf{\bf [this work]}\\ 
  \hline 
  \MaxSubarraySquare  
  & $O(n^3)$ \textsf{[trivial]} &  $O(n^{d+1})$ \textsf{[trivial]} \\
  on $n\times \cdots \times n$ arrays & $\Omega(n^3)$ \textsf{\bf [this work]} & $\Omega(n^{d+1})$ \textsf{\bf [this work]}\\ 
  \hline
  \WDepth & $O(N)$ \cite{chan2013klee}  
   & $O(N^{\nicefrac d 2})$ \cite{chan2013klee}\\
  on $N$ weighted boxes & $\Omega( N )$ \textsf{[trivial]} & $\Omega(N^{\nicefrac d 2})$ \textsf{\bf [this work]}\\
  \hline
\end{tabular}
\caption{\label{summary} Upper bounds and conditional lower bounds for the various problems studied. The bounds shown ignore subpolynomial factors.}
\end{table*} 
\egroup

\subsection{Related work on the problems}

In one dimension, the \MaxRect problem and \MaxSubarray problem are identical. The 1-D problem was first posed by Ulf Grenander
for pattern detection in images, and a linear time algorithm was found by Jay Kadane \cite{bentley1984programming}.

In two dimensions, Dobkin et al~\cite{dobkin1996computing, dobkin1994computing, maass1994efficient} studied the \MaxRect problem in the case where weights are $+1$ or $-1$ for its applications to computer graphics and machine learning. They presented the first $O(N^2 \log N)$ algorithm. More recently, Cort\'es et al~\cite{cortes2009bichromatic} studied the problem with arbitrary weights and they developed an algorithm with the same runtime applicable to many variants of the problem. An even faster algorithm was shown by Barbay et al.~\cite{barbay2014maximum} that runs in $O(N^2)$ time.

For higher dimensions, Barbay et al~\cite{barbay2014maximum} show a reduction to the related \WDepth problem which allows them to achieve runtime $O(N^d)$. Given $N$ axis-parallel rectangular weighted boxes, the \WDepth problem asks to find a point that maximizes the total weight of all boxes that contain it. Compared to the \MaxRect where we are given points and we aim to find the best box, in this problem, we are given boxes and the aim is to find the best point. The \WDepth problem is also related to Klee's measure problem\footnote{Klee's measure problem asks for the total volume of the union of $N$ axis-parallel boxes in $d$ dimensions.} which has a long line of research. All known algorithms for one problem can be adjusted to work for the other~\cite{chan2013klee}. The \WDepth problem was first solved in $O(N^{d/2} \log n)$ by Overmars and Yap~\cite{overmars1991new} and was improved to $O(N^{d/2})$ by Timothy M. Chan~\cite{chan2013klee} who gave a surprisingly simple divide and conquer algorithm.



A different line of work, studies the \MaxSubarray problem. Kadane's algorithm for the 1-dimensional problem can be generalized in higher dimensions for $d$-dimensional $n \times \cdots \times n$ arrays giving $O( n^{2d-1} )$ which implies an $O(n^3)$ algorithm when the array is a $n \times n$ matrix. Tamaki and Tokuyama~\cite{tamaki1998algorithms} gave a reduction of the 2-dimensional version of the problem to the distance product problem implying a $O\left( \frac {n^3} {2^{\Omega(\sqrt{\log n})}} \right)$ algorithm by using the latest algorithm for distance product by Ryan Williams~\cite{williams2014faster}. Tamaki and Tokuyama's reduction was further simplified by Tadao Takaoka~\cite{takaokaguy} who also gave a more practical algorithm whose expected time is close to quadratic for a wide range of random data.

\subsection{Our results and techniques}

Despite significant work on the \MaxRect and \MaxSubarray problems, it seems that there is a barrier in improving the best known algorithms for these problems by polynomial factors. 
Our results indicate that this barrier is inherent by showing connections to well-studied fundamental graph problems.
In particular, our first result states that there is no strongly subquadratic algorithm for the \MaxRect problem unless the \kClique problem can be solved in $O(n^{k-\eps})$ time, i.e. substantially faster than the currently best known algorithm. More precisely, we show the following:

\begin{restatable}{theorem}{genrect2d}
	\label{general_rect_hard_2d}
	For any constant $\eps>0$, an $O(N^{2-\eps})$ time algorithm for the \MaxRect problem on $N$ weighted points in the plane implies an $O(n^{\lceil\nicefrac 4 \eps\rceil-\eps})$ time algorithm for the \textsc{Max-Weight $\lceil\nicefrac 4 \eps\rceil$-Clique} problem on a weighted graph with $n$ vertices.
\end{restatable}

Our conditional lower bound generalizes to higher dimensions. Namely, we show that an $O(N^{d-\eps})$ time algorithm for points in $d$-dimensions implies an $O(n^{k-\eps})$ time algorithm for the \kClique problem for $k=\lceil\nicefrac{d^2}{\eps}\rceil$. This matches the best known algorithm~\cite{barbay2014maximum,chan2013klee} for any dimension up to subpolynomial factors. Therefore, because of our reduction, significant improvements in the runtime of the known upper bounds would imply a breakthrough algorithm for finding a $k$-clique of maximum weight in a graph.

To show this result, we embed an instance of the \kClique problem to the \MaxRect problem, by treating coordinates of the optimal rectangular box as base-$n$ numbers where digits correspond to nodes in the maximum-weight $k$-clique. In the construction, we place points
with appropriate weights so that the weight of any rectangular box corresponds to the weight of the clique it represents. We show that it is sufficient to use only $O(n^{\lceil \frac k d \rceil + 1})$ points in $d$-dimensions to represent all weighted $k$-cliques which gives us the required bound by choosing an appropriately large $k$.

We also study the special case of the \MaxRect problem in the plane where all points are arranged in a square grid, namely the \MaxSubarray problem. Our second result states that for $n \times n$ matrices, there is no strongly subcubic algorithm for the \MaxSubarray problem unless there exists a strongly subcubic algorithm for the \APSP problem. More precisely, we show that:

\begin{restatable}{theorem}{subarrayrestate} \label{subarray_2d}
	For any constant $\eps>0$, an $O(n^{3-\eps})$ time algorithm for the \MaxSubarray problem on $n \times n$ matrices implies an $O(n^{3 \ - \ \nicefrac{\eps}{10}})$ time algorithm for the \APSP problem.
\end{restatable}

We note that a reduction from \APSP problem to \MaxSubarray problem on $n \times n$ matrices was independently shown by Virginia Vassilevska Williams \cite{max_sum_submatrix}.

Combined with the fact that the \MaxSubarray problem reduces to the \APSP problem as shown in \cite{tamaki1998algorithms, takaokaguy}, our result implies that the two problems are equivalent, in the sense that any strongly subcubic algorithm for one would imply a strongly subcubic algorithm for the other. 

To extend our lower bound to higher dimensions, we need to make a stronger hardness assumption based on the \kClique problem. 
We show that an $O(n^{\nicefrac{3d}{2} \ -\ \eps})$ time algorithm for the \MaxSubarray problem in $d$-dimensions implies an $O(n^{k-\eps})$ time algorithm for the \kClique problem. To prove this result, we introduce the following intermediate problem: Given a graph $G$ find a maximum weight subgraph $H$ that is isomorphic to a clique on $2 d$ nodes without the edges of a matching (\CliqueMatching problem). This graph $H$ contains a large clique of size $\nicefrac {3d} 2$ as a minor and we show that this implies that no $O(n^{\nicefrac {3d} 2 \ -\  \eps})$ algorithms exist for the \CliqueMatching problem. We complete our proof by reducing the \CliqueMatching problem to the \MaxSubarray problem in $d$ dimensions. 

We note that the best known algorithm for the \MaxSubarray problem runs in $O(n^{2d-1})$ time and is based on Kadane's algorithm for the 1-dimensional problem. It remains an interesting open question to close this gap. To improve either the lower or upper bound, it is necessary to better understand the computational complexity of the \CliqueMatching problem.

Another related problem we consider is the \MaxSubarraySquare problem: Given an $n \times n$ matrix find a maximum subarray with sides of equal length. This problem and its higher dimensional generalization can be trivially solved in $O(n^{d+1})$ runtime
by enumerating over all possible combinations of the $d+1$ parameters, i.e. the side-length and the location of the hypercube.
 We give a matching lower bound based on hardness of the \kClique problem.

Finally, we adapt the reduction for Klee's measure problem shown by Timothy M Chan~\cite{chan2008slightly} to show a lower bound for the  \WDepth problem.

Our results are summarized in Table \ref{summary}, where we compare the current best upper bounds with the conditional lower bounds that we show.

The conditional hardness results presented above are for the 
variants of the problems where weights are arbitrary real numbers. We note that
all these bounds can be adapted to work for the case where
weights are either $+1$ or $-1$. In this case, we reduce the (unweighted) 
\CliqueDetect problem\footnote{Given a graph on $n$ vertices, the \CliqueDetect problem
asks whether a $k$-clique exists in the graph.} to each of these problems. 
The \CliqueDetect problem can be solved in $O(n^{\omega \lfloor k / 3 \rfloor + (k\text{ mod }3)})$ \cite{nevsetvril1985complexity} using fast matrix multiplication, where $\omega < 2.372864$ \cite{williams2012multiplying,le2014powers} is the fast matrix multiplication exponent. 
Without using fast matrix multiplication, it is not known whether a purely combinatorial algorithm exists that runs in $O(n^{k-\eps})$ time for any constant $\eps > 0$ and it is a longstanding graph problem. Our lower bounds can be adapted for the $+1\ /\ -1$ versions of the problems obtaining the same runtime exponents for combinatorial algorithms as in Table~\ref{summary}. Achieving better exponents for any of these problems would imply a breakthrough combinatorial algorithm for the \CliqueDetect problem.


There is a vast collection of problems in computation geometry for which conditional lower bounds are based on the assumption of \textsc{$3$-SUM} hardness, i.e. that the best known algorithm for the \textsc{$3$-SUM} problem\footnote{Given a set of integers, decide if there are $3$ integers that sum up to $0$.} can't be solved in time $O(n^{2-\eps})$. This line of research was initiated by \cite{gajentaan1995class} (see \cite{vassilevska2015hardness} for more references). Reducing \textsc{$3$-SUM} problem to the problems that we study seems hard if possible at all. Our work contributes to the list of interesting geometry problems for which hardness is shown from different assumptions.

\subsection{Hardness assumptions}
There is a long list of works showing conditional hardness for various problems based on the \APSP problem hardness assumption~\cite{roditty2004dynamic,williams2010subcubic,abboud2014popular,abboud2015subcubic,abboud2015matching}.
Among other results, \cite{williams2010subcubic} showed that deciding whether a weighted graph contains a triangle of 
negative weight is equivalent to the \APSP problem meaning that a strongly subcubic algorithm for the \NegTriangle problem
implies a strongly subcubic algorithm for the \APSP problem and the other way around. It is easy to show that the problem of computing the maximum weight triangle
in a graph is equivalent to the \NegTriangle problem (by inverting edge-weights of the graph and doing the binary search over the weight of the max-weight triangle).
Computing a max-weight triangle is a special case of the problem of computing a max-weight $k$-clique in a graph for a fixed integer $k$. This is a very well studied computational problem and despite serious efforts, the best known algorithm for this problem still runs in time $O(n^{k-o(1)})$, which matches the runtime of the trivial algorithm up to subpolynomial factors. The assumption that there is no $O(n^{k-\eps})$ time algorithm for this problem, has served as a basis for showing conditional hardness results for several problems on sequences
\cite{abboud2015if, abboud2014consequences}.

\section{Preliminaries}


\subsection{Problems studied in this work}

\begin{definition}[\MaxRect problem]
	\label{def_maxrect}
	Given $N$ weighted points (positive or negative) in $d\geq 2$ dimensions, what is the axis-aligned box 
	which maximizes the total weight of the points it contains? 
\end{definition}

\begin{definition}[\MaxSubarray problem]
	\label{def_maxsubarray}
Given a $d$-dimensional array $M$ with $n^d$ real-valued entries, find the $d$-dimensional subarray of $M$ which maximizes the sum of the elements it contains.
\end{definition}

\begin{definition}[\MaxSquare problem]
	\label{def_maxsquare}
Given a $d$-dimensional array $M$ with $n^d$ real-valued entries, find the $d$-dimensional square (hypercube) subarray of $M$, i.e. a rectangular box with all sides of equal length, which maximizes the sum of the elements it contains.
\end{definition}

\begin{definition}[\WDepth problem] \label{def_wdepth}
		Given a set of $N$ weighted axis-parallel boxes in $d$-dimensional space $\RR^d$, find a point $p \in \RR^d$ that maximizes the sum of the weights of the boxes containing $p$.
\end{definition}

\subsection{Hardness assumptions}

We use the hardness assumptions of the following problems.

\begin{definition}[\APSP problem]
	\label{def_apsp}
Given a weighted undirected graph $G=(V,E)$ such that $|V|=n$, find the shortest path between $u$ and $v$ for every $u,v \in V$.
\end{definition}

\begin{definition}[\NegTriangle problem]
	\label{def_negtriangle}
Given a weighted undirected graph $G=(V,E)$ such that $|V|=n$, output yes if there exists a triangle in the graph with negative total edge weight. 
\end{definition}

\begin{definition}[\kClique problem]
	\label{def_clique}
	Given an integer $k$ and a weighted graph $G=(V,E)$ with $n$ vertices, output the maximum total edge-weight of a $k$-clique in the graph. W.l.o.g. we assume that the graph is complete since otherwise we can set the weight of non-existent edges to be equal to a negative integer with large absolute value.
\end{definition}
For any fixed $k$, the best known algorithm for the \kClique problem runs in time $O(n^{k-o(1)})$.

In Sections \ref{sec_rect} and \ref{mr-hardness-d}, we use the following variant of the \kClique problem which can be shown to be equivalent to Definition \ref{def_clique}:
\begin{definition}[\kClique problem for $k$-partite graphs] \label{kpartite}
	Given an integer $k$ and a weighted $k$-partite graph $G=(V_1 \cup \ldots \cup V_k, \ E)$ with $k n$ vertices such that $|V_i|=n$ for all $i \in [k]$. Choose $k$ vertices $v_i \in V_i$ and consider total edge-weight of the $k$-clique induced by these vertices. Output the maximum total-edge weight of a clique in the graph.
\end{definition}

\paragraph*{Notation} For any integer $n$, we denote the set $\{1,2,\ldots,n\}$ by $[n]$. For a set $S$ and an integer $d$, we denote the set $\{(s_1,\ldots,s_d) \ | \ s_i \in S\}$ by $S^d$.

\section{Hardness of the \MaxRect problem} \label{sec_rect}

The goal of this section is to show a hardness result for the \MaxRect problem making the assumption of \kClique hardness. We will show the result directly for any constant number of dimensions.

\begin{restatable}{theorem}{genrect2}
	\label{general_rect_hard2}
	For any constants $\eps>0$ and $d$, an $O(N^{d-\eps})$ time algorithm for the \MaxRect problem on $N$ weighted points in $d$-dimensions
  implies an $O(n^{\lceil\nicefrac {d^2} \eps \rceil-\eps})$ time algorithm for the \textsc{Max-Weight 
  $\lceil \nicefrac {d^2} \eps \rceil$-Clique} problem on a weighted graph with $n$ vertices.
\end{restatable}

We set $k = \lceil \frac d \eps \rceil$. To prove the theorem, we will construct an instance of the \MaxRect problem whose answer computes a {max-weight $d k$-clique} in a $(d \times k)$-partite weighted graph G with $n$ nodes in each of its parts. 
The \textsc{Max-Weight $d k$-Clique} problem on general graphs reduces to this case since we can create $d \times k$ copies of the nodes and connect nodes among different parts with edge-weights as in the original graph. 
  
The instance of the \MaxRect problem will consist of $N=O(n^{k+1})$ points with integer coordinates $\{-n^k,...,n^k\}^d$. For such an instance the required runtime for the \MaxRect problem, from the theorem statement, would imply that the maximum weight $d k$-clique can be computed in 
$$O\left(N^{d - \eps}\right) = O\left(N^{d(1 \ - \ \nicefrac 1 k)}\right) = O\left(n^{d(k \ - \ \nicefrac 1 k)}\right) = O\left(n^{d k - \eps}\right).$$

To perform the reduction we introduce the following intermediate problem:
\begin{definition}[\textsc{Restricted Rectangle} problem]
Given $N=\Omega(n^k)$ weighted points in an $\{-n^k,...,n^k\}^d$-grid, compute a rectangular box of a restricted form that maximizes the weight of its enclosed points. The rectangular box $\prod_{i=1}^d [-x'_i,x_i]$ must satisfy the following conditions:
\begin{enumerate}
  \item \label{resrect:cond1} Both $\vec x, \vec x' \in \{0,...,n^k-1\}^d$, and
  \item \label{resrect:cond2} Treating each coordinate $x_i$ as a $k$-digit integer $(x_{i1}x_{i2}...x_{ik})_n$ in base $n$, i.e. $x_i = \sum_{j=1}^k x_{ij} n^{k-j}$, we must have $\vec x' = (\overline {x_d}, \overline {x_1}, \overline {x_2},...,\overline {x_{d-1}})$, where for an integer $z = (z_1 z_2 ... z_k)_n \in  \{0,...,n^k-1\}$, we denote by $\overline {z} = (z_k  ... z_2 z_1)_n$ the integer that has all the digits reversed.
\end{enumerate}
\end{definition}
\noindent We show that the \textsc{Restricted Rectangle} problem reduces to the \MaxRect problem.

\subsection{\textsc{Restricted Rectangle} $\Rightarrow$ \MaxRect}

Consider an instance of the \textsc{Restricted Rectangle} problem. We can convert it to an instance of the \MaxRect problem by introducing several additional points. Let $C$ be a number greater than twice the sum of absolute values of all weights of the given points. We know that the solution to any rectangular box must have weight in $(-\nicefrac C 2,\nicefrac C 2)$.

The conditions of the \textsc{Restricted Rectangle} require that the rectangular box must contain the origin $\vec 0$. To satisfy that 
we introduce a point with weight $C$ at the origin. This forces the optimal rectangle to contain the origin since any rectangle that doesn't include this point gets weight strictly less than $C$.

The integrality constraint is satisfied since all points in the instance have integer coordinates so without loss of generality the optimal rectangle in the \MaxRect problem will also have integer coordinates.

Finally, we can force $x'_2 = \overline {x_1}$, by adding for each $x_1 \in \{0,...,n^k-1\}$ the 4 points:
\begin{itemize}
 \item $(x_1,-\overline {x_1},0,0,..,0)$ with weight $C$
 \item $(x_1+1,-\overline {x_1},0,0,..,0)$ with weight $-C$
 \item $(x_1,-\overline {x_1}-1,0,0,..,0)$ with weight $-C$
 \item $(x_1+1,-\overline {x_1}-1,0,0,..,0)$ with weight $C$
\end{itemize}
This creates $4 n^{k}$ points and adds weight $C$ to any rectangle with $x'_2 = \overline {x_1}$ without affecting any of the others. Working similarly for $x_2...,x_d$ we can force that the optimal solution satisfies the constraint that $\vec x' = (\overline {x_d}, \overline {x_1}, \overline {x_2},...,\overline {x_{d-1}})$. 

If $x$ and $x'$ satisfy the constraints of the \textsc{Restricted Rectangle} problem, we collect total weight at least $(d+1) C - \frac C 2 = (d+\frac 1 2) C$.
If at least one of the constraints is not satisfied, we receive weight strictly less than $(d+\frac 1 2) C$. Thus, the optimal rectangular box for the \MaxRect problem satisfies all the necessary constraints and coincides with the optimal rectangular box for the \textsc{Restricted Rectangle} problem. The total number of points is still $O(N)$ since $N=\Omega(n^k)$ and we added $O(n^k)$ points.

\subsection{\textsc{Max-Weight $(d \times k)$-Partite Clique} $\Rightarrow$ \textsc{Restricted Rectangle}}

Consider a $(d \times k)$-partite weighted graph $G$. We label each of its parts as $P_{ij}$ for $i \in [d]$ and $j \in [k]$.
We associate each $d k$-clique of the graph $G$ with a corresponding rectangular box in the \textsc{Restricted Rectangle} problem. In particular, for a rectangular box defined by a point $\vec x \in \{0,...,n^k-1\}^d$, each $x_{ij}$, i.e. the $j$-th most significant digit of $x_i$ in the base $n$ representation, corresponds to the index of the node in part $P_{ij}$ (0-indexed).

We now create an instance by adding points so that the total weight of every rectangular box satisfying the conditions of the \textsc{Restricted Rectangle} problem is equal to the weight of its corresponding $d k$ clique. To do that we need to take into account the weights of all the edges. We can easily take care of edges between parts $P_{11},P_{12},...,P_{1k}$ of the graph by adding the following points for each $x_1  \in \{0,...,n^k-1\}$.
\begin{itemize}
 \item $(x_1,0,0,0,..,0)$ with weight $W(x_1)$ equal to the weight of the $k$-clique $x_{11},x_{12},...,x_{1k}$ in parts $P_{11},P_{12},...,P_{1k}$
 \item $(x_1+1,0,0,0,..,0)$ with weight $-W(x_1)$
\end{itemize}
This creates $2 n^{k}$ points and adds weight $W(x_1)$ to any rectangle whose first coordinate matches $x_1$ without affecting any of the others. We work similarly for every coordinate $i$ from $2$ through $d$ accounting for the weight of all edges between parts $P_{ia}$ and $P_{ib}$ for all $i \in [d]$ and $a \neq b \in [k]$.  To take into account the additional edges, we show how to add edges between parts $P_{1a}$ and $P_{2b}$. For all $x_{1} \in n^{k-a} \{0,...,n^a-1\}$ and $x_{2} \in n^{k-b} \{0,...,n^b-1\}$ we add the points:
\begin{itemize}
 \item $(x_1,x_2,0,0,..,0)$ with weight $w$ equal to the weight of the edge between nodes $x_{1a}$ and $x_{2b}$ in parts $P_{1a}$ and $P_{2b}$.
 \item $(x_1+n^{k-a},x_2,0,0,..,0)$ with weight $-w$
 \item $(x_1,x_2+n^{k-b},0,0,..,0)$ with weight $-w$
 \item $(x_1+n^{k-a},x_2+n^{k-b},0,0,..,0)$ with weight $w$
\end{itemize}
This adds weight equal to the weight of the edge between nodes $x_{1a}$ and $x_{2b}$ in parts $P_{1a}$ and $P_{2b}$ for any rectangle with corner $\vec x$. This creates $O(n^{a+b})$ points. This number becomes too large if $a+b > k+1$. However, if this is the case we can instead apply the same construction in the part of the space where the numbers $x_1$ and $x_2$ appear reversed, i.e. by working with $x'_2 = \overline {x_1}$ and $x'_3 = \overline {x_2}$. For all $x'_{2} \in n^{a-1} \{0,...,n^{k+1-a}-1\}$ and $x'_{3} \in n^{b-1} \{0,...,n^{k+1-b}-1\}$ we add the points:
\begin{itemize}
 \item $(0,-x'_2,-x'_3,0,0,..,0)$ with weight $w$ equal to the weight of the edge between nodes $x'_{2(k+1-a)}$ and $x'_{3(k+1-b)}$ in parts $P_{1a}$ and $P_{2b}$.
 \item $(0,-x'_2-n^{a-1},-x'_3,0,..,0)$ with weight $-w$
 \item $(0,-x'_2,-x'_3-n^{b-1},0,..,0)$ with weight $-w$
 \item $(0,-x'_2-n^{a-1},-x'_3-n^{b-1},0,..,0)$ with weight $w$
\end{itemize}
This produces the identical effect as above creating $O(n^{2k+2-a-b})$ rectangles. If $a+b \ge k+1$ this adds at most $O(n^{k+1})$ points as desired. We add edges between any other 2 parts $P_{i,\cdot}$ and $P_{i',\cdot}$ by performing a similar construction as above.

The overall number of points in the instance is $O(n^{k+1})$ and this completes the proof of the theorem.

\section{Hardness for \MaxSubarray in $2$ dimensions}\label{sec_mr_hardness_2}
	In this section our goal is to show that, if we can solve the \MaxSubarray problem on a matrix of size $n \times n$ in time $O(n^{3-\eps})$, then we can solve the \NegTriangle problem in time $O(n^{3-\eps})$ on $n$ vertex graphs. It is known that a $O(n^{3-\eps})$ time algorithm for the \NegTriangle implies a $O(n^{3 \ - \ \nicefrac{\eps}{10}})$ time algorithm for the \APSP  problem \cite{williams2010subcubic}. Combining our reduction with the latter one, we obtain Theorem \ref{subarray_2d} from the introduction, which we restate here:

	\subarrayrestate*
	
	The generalization of this statement can be found in Section~\ref{mr-hardness-d}. Here we prove $2$-dimensional case first because the argument is shorter.

	Clearly, the \NegTriangle problem in equivalent to the \PosTriangle problem. In the remainder of this section we therefore reduce the problem of detecting whether a graph has a positive triangle to the \MaxSubarray problem. 

	We need the following intermediate problem:
	\begin{definition}[\MfComb] \label{mfcomb}
		Given a matrix $B \in \RR^{m \times m}$, output
		$$
			\max_{i, i', j, j' \in [m] \ : \ i\leq i' \text{ and }j \leq j'}B[i,j]+B[i',j']-B[i,j']-B[i',j].
		$$
	\end{definition}

	Our reduction consists of two steps:
	\begin{enumerate}
		\item Reduce the \PosTriangle problem on $n$ vertex graph to the \MfComb problem on $2n \times 2n$ matrix.
		\item Reduce the \MfComb problem on $n \times n$ matrix to the \MaxSubarray matrix of size $n \times n$.
	\end{enumerate}
	
	\subsection{\PosTriangle $\Rightarrow$ \MfComb}
		Let $A$ be the weighted adjacency matrix of size $n \times n$ of the graph and let $M$ be the largest absolute value of an entry in $A$. Let $M':=10M$ and $M'':=100M$.
		We define matrix $D \in \RR^{n \times n}:$
		$$
			D_{i,j}=
			\begin{cases}
				M'+M'' & \text{ if }i=j;\\
				M'' & \text{ otherwise}.
			\end{cases}
		$$

		We define matrix $B \in \RR^{2n \times 2n}:$
		$$
			B:=
			\begin{bmatrix}
				A     & -A^T \\
				-A^T  & D
			\end{bmatrix}.
		$$
		
		The reduction follows from the following lemma.
		\begin{lemma}
			Let $X$ be the weight of the max-weight rectangle in the graph corresponding to the adjacency matrix $A$.
			Let $Y$ be the output of the \MfComb algorithm when run on matrix $B$.
			The following equality holds:
			$$
				Y=X+M'+M''.
			$$
		\end{lemma}
		\begin{proof}
			Consider integers $i,j,i',j'$ that achieve a maximum in the \MfComb instance as per Definition \ref{mfcomb}.
			Our first claim is that $i,j\leq n$ and $i',j'\geq n+1$. If this is not true, we do not collect the weight $M''$ and the largest output that we can get is $\leq 4M'\leq 9M''/10$. Note that we can easily achieve a larger output with $i=j=1$ and $i'=j'=n+1$.
			
			Our second claim is that $i'=j'$. If this is not so, we do not collect the weight $M'$ and the largest output that we can get is $M''+4M\leq M''+M'/2$. Note that we can easily achieve a larger output with $i=j=1$ and $i'=j'=n+1$. Thus, we can denote $i'=j'=k+n$.
			
			Now, by the construction of $B$, we have
			$$
				B[i,j]+B[i',j']-B[i,j']-B[i',j]=A[i,j]+A[j,k]+A[k,i]+M'+M''.
			$$
			We get the equality we need.
		\end{proof}

	\subsection{\MfComb $\Rightarrow$ \MaxSubarray}
		Let $A' \in \RR^{(n+1) \times (n+1)}$ be a matrix defined by $A'[i,j]=A[i-1,j-1]$ if $i,j\geq 2$ and $A'[i,j]=0$ otherwise.
		
		Let $C \in \RR^{n \times n}$ be a matrix defined by $C[i,j]=A'[i,j]+A'[i+1,j+1]-A'[i,j+1]-A'[i+1,j]$.
		
		The reduction follows from the claim that the output of the \MaxSubarray on $C$ is equal to the output of the $\MfComb$ on $A'$.
		The claim follows from the following equality:
		$$
			\sum_{i=i'}^{i''}\sum_{j=j'}^{j''}C[i,j] \ = \ A'[i''+1,j''+1]+A'[i',j']-A'[i''+1,j']-A'[i',j''+1].
		$$

\section{Hardness for \MR for arbitrary number of dimensions}\label{mr-hardness-d}
We can extend the ideas used in the hardness proof of Theorem \ref{subarray_2d}, to prove the following theorem for the \MR problem on $d$ dimensional arrays.

\begin{restatable}{theorem}{subarray_HD}\label{subarray_HD}
For any constant $\eps>0$, an $O\left(n^{d+\lfloor d/2\rfloor-\eps}\right)$ time algorithm for the \MR problem on $d$-dimensional array, implies an $O\left(n^{d+\lfloor d/2\rfloor-\eps}\right)$ time algorithm for the \wClique{d+\lfloor d/2\rfloor} problem.
\end{restatable}

To prove the theorem, we introduce some notation and define some intermediate problems which will be helpful in modularizing the reduction. We will also be using the notation introduced here in Section \ref{sec_square}.
	\begin{definition}[$d$-Tuple]
		$i$ is $d$-tuple if $i=(i_1,\ldots,i_d)$ for some integers $i_1,\ldots,i_d$.
	\end{definition}

	\paragraph*{Notation}
		Let $i$ be the $d$-tuple $(i_1,\ldots,i_d)$ and $\Delta$ be an integer.
		We denote the $d$-tuple $(\Delta\cdot i_1,\ldots,\Delta \cdot i_d)$ by $\Delta \cdot i$.
		Let $j$ be the $d$-tuple $j=(j_1,\ldots,j_d)$. We denote the $d$-tuple $(i_1+j_1,\ldots,i_d+j_d)$ by $i+j$.
		For $d$-tuple $i=(i_1,\ldots,i_d)$, we denote sum $|i_1|+\ldots+|i_d|$ by $\|i\|_1$. If $i$ is binary, $\|i\|_1$ denotes the number of ones in $i$.
		$j^t$ is the binary vector with only one entry equal to $1$: $j^t_t=1$. 
		That is, the $t$-th entry of $j^t$ is equal to $1$. For $d$-tuple $i$, we define \emph{type} $type(i)$ of $i$ as follows. 
		$type(i)$ is a binary vector such that for every $t \in [d]$, $type(i)_t=0$ iff $i_t<0$.
		Given two $d$-tuples $i=(i_1,\ldots,i_d)$ and $j=(j_1,\ldots,j_d)$, we denote $d$-tuple $(i_1\cdot j_1, \ldots, i_d\cdot j_d)$ by $i \times j$.

	\begin{definition}[$d$-Dimensional Array]
		We call $A$ an array in $d$ dimensions of side-length $n$ if it satisfies the following properties.
		\begin{itemize}
			\item $A$ contains $n^d$ real valued entries.
			\item $A[i]=A[i_1,\ldots,i_d]$ is the entry in $A$ corresponding to $d$-tuple $i=(i_1,\ldots,i_d) \in [n]^d$.
		\end{itemize}
	\end{definition}

	\begin{definition}[Boolean Cube]
		Let $B_d:=\{0,1\}^d$ be a set consisting of all $2^d$ binary $d$-tuples. We call it a Boolean cube in $d$ dimensions.
	\end{definition}

	\begin{definition}[Central $d$-Dimensional Array]
		We call $A$ a central array in $d$ dimensions of side-length $2n+1$ if it satisfies the following properties.
		\begin{itemize}
			\item $A$ contains $(2n+1)^d$ real valued entries.
			\item $A[i]=A[i_1,\ldots,i_d]$ is the entry in $A$ corresponding to $d$-tuple 
			$$
				i=(i_1,\ldots,i_d) \in \{-n,-n+1,\ldots,n-1,n\}^d.
			$$
		\end{itemize}
	\end{definition}

	\begin{definition}[\MaxSub{2k}] \label{subgraph_problem}
		We are given integer $k$ and weighted $2k$-partite graph $G=(V_1 \cup V_2 \ldots V_k \cup V_1' \cup V_2'\ldots V_k', \ E)$ with $2kn$ vertices. $|V_i|=|V_i'|=n$ for all $i \in [k]$. Choose $2k$ vertices $v_i \in V_i$, $v_i' \in V_i'$ and define
		$$
			W:=\sum_{i \in [k]} \sum_{j \in [k]\setminus\{i\}} w(v_i,v_j')+w(v_i,v_j)+w(v_i',v_j').
		$$
		$w(u,v)$ denotes the weight of edge $(u,v)$.
		In other words, $W$ is equal to the total edge-weight of $2k$-clique induced by $2k$ vertices $v_i$, $v_j'$ from which we subtract weight contributed by $k$ edges $(v_i,v_i')$. 
		The computation problem is to output maximum $W$ that we can obtain by choosing the $2k$ vertices.
	\end{definition}
	
	The trivial algorithm solves this problem in time $O(n^{2k})$. We can improve the runtime to $O(n^{2k-1})$. 
	Below we show that we cannot get runtime $O\left(n^{k+\lfloor k/2\rfloor-\Omega(1)}\right)$ unless we get a
	much faster algorithm for the \textsc{Max-Weight Clique} problem than what currently is known.
	
	\begin{definition}[\CMRS problem] \label{cmrs}
		Let $A$ be a central array in $d$ dimensions of side-length $2n+1$.
		We must output
		$$
			\max_{\substack{i \in [n]^d,\ \delta \in [2n]^d\\ \text{s.t. }\delta_1-i_1,\ldots,\delta_d-i_d\ge 0}}\sum_{j \in B_d}A[-i+\delta\times j].
		$$
	\end{definition}
	
	\begin{definition}[\CMRC problem] \label{cmrc}
		Let $A$ be a central array in $d$ dimensions of side-length $2n+1$.
		We must output
		$$
			\max_{\substack{i \in [n]^d,\ \delta \in [2n]^d\\ \text{s.t. }\delta_1-i_1,\ldots,\delta_d-i_d\ge 0}}\sum_{j \in B_d}(-1)^{\|j\|_1}\cdot A[-i+\delta\times j].
		$$
	\end{definition}
	
	\begin{definition}[\MRC problem] \label{mrc}
		Let $A$ be an array in $d$ dimensions of side-length $2n+1$.
		We must output
		$$
			\max_{i \in [n]^d,\ \delta \in [2n]^d}\sum_{j \in B_d}(-1)^{\|j\|_1}\cdot A[-i+\delta\times j].
		$$
	\end{definition}
	
	\begin{definition}[\MR problem] \label{mr}
		Let $A$ be an array in $d$ dimensions of side-length $n$.
		We must output
		$$
			\max_{i,\delta \in [n]^d}\sum_{i_1\le k_1\le \delta_1}\ldots \sum_{i_d\le k_d\le \delta_d}A[k_1,\ldots,k_d].
		$$
	\end{definition}
	
	Our goal is to show that, if we can solve \MR problem in time $O\left(n^{d+\lfloor d/2\rfloor-\eps}\right)$ for some $\eps>0$ on $d$-dimensional array, then we can solve \wClique{d+\lfloor d/2\rfloor} problem in time $O\left(n^{d+\lfloor d/2\rfloor-\eps}\right)$. Below, whenever, we refer to an array, it has $d$ dimensions.
	
	We will achieve this goal via a series of reductions:
	\begin{enumerate}
		\item Reduce \wClique{d+\lfloor d/2\rfloor} on $\left(d+\lfloor d/2\rfloor\right)n$ vertex graph to \MaxSub{2d} on $2dn$ vertex graph.
		\item Reduce \MaxSub{2d} problem on $2dn$ vertex graph to \CMRS on array with side-length $2dn+1$.
		\item Reduce \CMRS on array with side-length $2n+1$ to \CMRC on array with side-length $2n+1$.
		\item Reduce \CMRC on array with side-length $2n+1$ to \MRC on array with side-length $2n+1$.
		\item Reduce \MRC on array with side-length $2n+1$ to \MR on array with side-length $2n$.
	\end{enumerate}

	We can check that this series of reductions is sufficient for our goal. (For this, remember our assumption that $d=O(1)$.) Also, all reductions can be performed in time $O(n^d)$.
	
	\begin{remark}
		It is possible to show that there is no $O\left(n^{3d/2-\eps}\right)$ time algorithm for the \MR problem unless we have a much faster algorithm for \textsc{Max-Weight Clique} problem. The proof of this lower bound, however, is more complicated, and we omit it here.
	\end{remark}

	\subsection{\wClique{d+\lfloor d/2\rfloor} $\Rightarrow$ \MaxSub{2d}}
		Given an instance of the \wClique{d+\lfloor d/2\rfloor} problem on $\left(d+\lfloor d/2\rfloor\right)$-partite graph 
		$$
			G=(V_1 \cup \ldots \cup V_{d+\lfloor d/2\rfloor},E),
		$$
		we transform it into an instance of the \MaxSub{2d} on graph
		$$
			G'=(V_1 \cup \ldots \cup V_d \cup V_1' \cup \ldots \cup V_d',E')
		$$ 
		as follows. We build $G'$ out of $G$ in three steps.
		
		\paragraph*{Step 1} $G'$ is the same as $G$, except that we rename $V_{i+d}$ as $V_i'$ for $i = 1, \ldots, \lfloor d/2\rfloor$. Clearly, the max-weight clique in $G'$ is of the same weight as the max-weight clique in $G$.
		
		\paragraph*{Step 2} For $i=1,\ldots,\lfloor d/2\rfloor$, we do the following. We add a set of vertices 
		$$
			V_{i+\lfloor d/2\rfloor}':=\{v' \ : \ v \in V_i'\}
		$$ 
		to $G'$.
		For every $v \in V_i'$ and $u \in V_{i+\lfloor d/2\rfloor}'$, we set the weight of the edge $(v,u)$ as follows:
		$$
			w(v,u):=
			\begin{cases}
				0, & \text{ if }u=v';\\
				-M & \text{ otherwise},
			\end{cases}
		$$
		where $M=100\cdot d^{10} \cdot W$ and $W$ is the largest absolute value of the edge weight in $G$. $M$ is chosen to be a sufficiently large positive value.
		For every $u \in V_i$ and $v' \in V_{i+\lfloor d/2\rfloor}'$, we set the weight of the edge $(u,v')$ to be equal to the weight of the edge $(u,v)$: $w(u,v'):=w(u,v)$.
		We set all unspecified edge weights to be equal to $0$.
		
		\paragraph*{Step 3} If $2\lfloor d/2\rfloor<d$, we add a set of vertices $V_d'$ to $G'$, and we set all unspecified edges to have weight $0$.
		
		The correctness of the reduction follows from the following theorem.
		\begin{theorem}
			The maximum weight of $\left(d+\lfloor d/2\rfloor\right)$-clique in $G$ is equal to the maximum weight $2d$-subgraph of $G'$ (see Definition \ref{subgraph_problem}).
		\end{theorem}
		\begin{proof}
			Fix any $i$ in $\{1,\ldots,\lfloor d/2\rfloor\}$.
			If, when choosing maximum weight $2d$-subgraph of $G'$, we pick vertex $v \in V_i'$, then we must pick vertex $v'$ from $V_{i+\lfloor d/2\rfloor}'$ since, otherwise, we would collect cost $-M$ by the construction.
			Suppose we pick $v \in V_i'$ and $u \in V_i$. Since we have to pick $v'$ from $V_{i+\lfloor d/2\rfloor}'$ and since the weight of $(u,v)$ is equal to the weight of $(u,v')$, we must collect the weight of the edge $(u,v)$. Now the correctness of the claim follows from Definition \ref{subgraph_problem}.
		\end{proof}

	\subsection{\MaxSub{2d} $\Rightarrow$ \CMRS}
		Given a $2d$-partite graph 
		$$
			G=(V_1 \cup \ldots \cup V_d \cup V_1' \cup \ldots \cup V_d',E),
		$$
		we construct array $A$ with side-length $2n+1$ as follows.
		Let $i \in \{-n,\ldots,n\}^d$ be a $d$-tuple.
		We set $A[i]=-M'$, if there exists $r \in [d]$ such that $i_r=0$. We set $M'=100^{10d}\cdot W'$, where $W'$ is the largest absolute value among the edge weights in $G$. $M'$ is chosen to be a sufficiently large positive value. 
		We choose $d$ vertices $v_1,\ldots,v_d$ from $G$ as follows. If $i_k<0$, we set $v_k$ to be the $(-i_k)$-th
		vertex from set $V_k$. If $i_k>0$, we set $v_k$ to be the $i_k$-th vertex from set $V_k'$. We set $A[i]$ to be equal to the total weight of $d$-clique spanned by vertices $v_1,\ldots,v_d$.
		
		We need the following lemma.
		\begin{lemma} \label{max_subg}
			Fix $i \in [n]^d$ and $\delta \in [2n]^d$ such that $n \geq \delta_r-i_r>0$ for all $r \in [d]$.
			For every $r \in [d]$, set $u_r$ to be the $i_r$-th vertex from $V_r$ and $u_r'$ to be the $(\delta_r-i_r)$-th
			vertex from $V_r'$. Then
			$$
				\sum_{j \in B_d}A[-i+\delta \times j]=2^{d-2}\cdot w,
			$$
			where $w$ is the total weight of $2d$-subgraph spanned by vertices $u_1,\ldots,u_d,u_1',\ldots,u_d'$.
		\end{lemma}
		\begin{proof}
			Follows from Definition \ref{subgraph_problem} and the construction of array $A$.
		\end{proof}
		
		We observe that, as we maximize over $d$-tuples $i$ and $\delta$ (as per Definition \ref{cmrs}), we never choose $i$ and $\delta$ such that there exists $r$ with $\delta_r-i_r=0$ so as to not collect $-M'$. Also, we see that, as we maximize over all $i$ and $\delta$, we maximize over all $2d$-subgraphs by Lemma \ref{max_subg}. The output of \CMRS problem on $A$ is therefore equal to the maximum weight of a $2d$-subgraph in $G$ multiplied by $2^{d-2}$. This finishes the description of the reduction.

	\subsection{\CMRS $\Rightarrow$ \CMRC}
		Let $A$ be the input array for the \CMRS problem.
		We construct $A'$ as follows.
		For every $i \in \{-n,\ldots,n\}^d$:
		$$
			A'[i]:=
			\begin{cases}
				A[i] & \text{ if }|\{r \ : \ i_r \geq 0 \}|\text{ is even},\\
				-A[i] & \text{ otherwise}.
			\end{cases}
		$$
		Our claim is that the output of the \CMRC on $A'$ is equal to the output of \CMRS on $A$. This follows by the definitions of the both problems.
	
	\subsection{\CMRC $\Rightarrow$ \MRC}
		Let $A$ be the input array for the \CMRC problem.
		Let $W''$ be the largest absolute value of an entry in $A$.
		We define $M''=100^{10d}\cdot W''$ to be large enough positive value.
		
		We define $A'$ as follows. First we set $A':=A$. Then, for every $d$-tuple $i$ with $i_r<0$ for all $r \in [d]$, we increase $A'[i]$ by $M''$.
		
		The reduction follows from the following lemma.
		\begin{lemma}
			Let $X$ be the output of the \CMRC on input $A$.
			Let $X'$ be the output of the \MRC on input $A'$.
			Then equality $X'=X+M''$ holds.
		\end{lemma}
		\begin{proof}
			Consider \MRC on input $A'$.
			
			We claim that a maximum cannot be achieved for $d$-tuples $i$ and $\delta$ such that there exists $r \in [d]$ with $\delta_r-i_r<0$. Suppose that there are such $i$ and $\delta$ that achieve a maximum. By the construction of $A'$, and because $\delta_r-i_r<0$, all values $M''$ that we collect will cancel out among themselves. We will then be left with value, at most, $X'\leq |B_d|W''\leq {1 \over 10}M''$. We can, however, achieve a value of at least ${9 \over 10}M''>{1 \over 10}M''$ by setting $i_k=-n$ and $\delta_k=0$ for all $k \in [d]$.
			
			By the discussion in the previous paragraph, a maximum must be achieved for $i$ and $\delta$ such that $\delta_r-i_r\geq 0$ for all $r \in [d]$.
			Now this is exactly the condition that we impose on $i$ and $\delta$ in the statement of the \CMRC problem. By the construction of $A'$, we get equality $X'=X+M''$.
		\end{proof}
		
	\subsection{\MRC $\Rightarrow$ \MR}\label{mrctomr}
	Let $A$ be the input $d$-dimensional array with side-length $2n+1$ to the \MRC problem. Given $A$, we produce $d$-dimensional array $A'$ of side-length $2n$ such that the output of the \MR problem on $A'$ is equal to the output of the \MRC problem on $A$.
		We construct $A'$ as follows. For every $d$-tuple $i \in [2n]^d$, we set
		$$
			A'[i]=\sum_{j \in B_d}(-1)^{\|j\|_1}\cdot A[i+j].
		$$
		We can check equality
		\begin{equation} \label{mcmr}
			 \sum_{i_1\le k_1\le i_1+\delta_1}\ldots \sum_{i_d\le k_d\le i_d+\delta_d}A'[k_1,\ldots,k_d] 
			=  \sum_{j \in B_d}(-1)^{\|j\|_1}\cdot A[i+(\delta+\underline{1})\times j].
		\end{equation}
		where $\underline{1}$ is the $d$-tuple $(1,\ldots,1)$.
		In the \MR problem, we maximize l.h.s. of \eqref{mcmr} over $d$-tuples $i$ and $d$-tuples $\delta \in [2n]^d$. In the \MRC problem, we maximize r.h.s. of \eqref{mcmr} over $d$-tuples $i$ and $d$-tuples $\delta \in [2n]^d$. The reduction follows from the definitions of the computational problems.	
	

\section{Hardness for \MCube problem} \label{sec_square}

When the side-lengths of the subarray we are looking for are restricted to be equal, the problem becomes slightly easier and there exists a $O\left(n^{d+1}\right)$ algorithm for solving it. In this section, we show a matching lower bound for the \MCube problem.

\begin{restatable}{theorem}{square_hardness}\label{square_hardness}
For any constant $\eps>0$, an $O\left(n^{d+1-\eps}\right)$ time algorithm for the \MCube problem on a $d$-dimensional array implies an $O\left(n^{d+1-\eps}\right)$ time algorithm for the \wClique{d+1} problem.
\end{restatable}

	To prove Theorem \ref{square_hardness}, we define some intermediate problems which will be helpful in modularizing the reduction.
\begin{definition}[\CMaxSum problem]
	Let $A$ be a central array in $d$ dimensions of side-length $2n+1$.
	We must output
	$$
		\max_{\substack{i \in [n]^d,\ \Delta \in [2n]\\ \text{s.t. }\Delta-i_1,\ldots,\Delta-i_d\ge 0}}\sum_{j \in B_d}A[-i+\Delta\cdot j].
	$$
\end{definition}

\begin{definition}[\CMC problem]
	Let $A$ be a central array in $d$ dimensions of side-length $2n+1$.
	We must output
	$$
		\max_{\substack{i \in [n]^d,\ \Delta \in [2n]\\ \text{s.t. }\Delta-i_1,\ldots,\Delta-i_d\ge 0}}\sum_{j \in B_d}(-1)^{\|j\|_1}\cdot A[-i+\Delta\cdot j].
	$$
\end{definition}

\begin{definition}[\MComb problem]
	Let $A$ be an array in $d$ dimensions of side-length $n$.
	We must output
	$$
		\max_{i \in [n]^d,\ \Delta \in [n]}\sum_{j \in B_d}(-1)^{\|j\|_1}\cdot A[i+\Delta\cdot j].
	$$
\end{definition}

\begin{definition}[\MCube problem]
	Let $A$ be an array in $d$ dimensions of side-length $n$.
	We must output
	$$
		\max_{i \in [n]^d,\ \Delta \in \{0,\ldots,n-1\}}\sum_{i_1\le k_1\le i_1+\Delta}\ldots \sum_{i_d\le k_d\le i_d+\Delta}A[k_1,\ldots,k_d].
	$$
\end{definition}

We note that there is a simple algorithm for \MCube problem that runs in time $O(n^{d+1})$.

	Our goal is to show that, if we can solve \MaxSubarraySquare in time $O(n^{d+1-\eps})$ for some $\eps>0$ on $d$-dimensional array, where $d\geq 3$ is a constant, then we can solve \wClique{d+1} in time $O(n^{d+1-\eps})$. Below, whenever we refer to an array, it has $d$ dimensions.

	We will show this by a series of reductions:
	\begin{enumerate}
		\item Reduce \wClique{d+1} on $n$ vertex graph to \CMaxSum problem on array with side-length $2dn+1$.
		\item Reduce \CMaxSum problem on array with side-length $2n+1$ to \CMC problem on array with side-length $2n+1$.
		\item Reduce \CMC problem on array with side-length $2n+1$ to \MComb problem on array with side-length $2n+1$.
		\item Reduce \MComb problem on array with side-length $n$ to \MCube problem on array with side-length $n-1$.
	\end{enumerate}

	We can check that this series of reductions is sufficient for our goal. All reductions can be performed in time $O(n^d)$.

	\subsection{\wClique{d+1} $\Rightarrow$ \CMaxSum}
		Given a weighted graph $G=(V,E)$ on $n$ vertices, our goal is to produce a $d$-dimensional array $A$ with side-length $2dn+1$ so that the following holds. If we solve the \CMaxSum problem on $A$, we can infer the maximum total edge-weight of $(d+1)$-clique in $G$ in constant time.

		We set $c'$ to be equal to the maximum absolute value of the edge-weights in $G$. We set $c:=100|V|^4c'$, which is much larger than the total edge weight of the graph. We define the following $d$-dimensional array $D$ of side-length $n$. For every $d$-tuple $i \in [n]^d$, we set $D[i]$ by the following rules.
		\begin{enumerate}
			\item If there are $r\neq t \in [d]$ such that $i_r=i_t$, set $D[i]=-c$.
			\item Otherwise, set $D[i]$ to be equal to the total edge weight of $d$-clique with vertices $i_1,\ldots,i_d$.
		\end{enumerate}

		Using array $D$, we construct array $A$ in the following way:
		\begin{itemize}
			\item Initially, set every entry of $A$ to be equal to $-c$.
			\item For every $i \in [n]^d$, set $A[-i]=D[i]$.
			\item For every $t \in [d]$ and $i \in [n]^d$, set 
				\begin{equation} \label{defa}
					A[-(i-i_t\cdot j^t)+\|i\|_1\cdot j^t]=D[i].
				\end{equation}
		\end{itemize}

		The following theorem completes our reduction.
		\begin{theorem}
			Let $M_A$ be the output of the \CMaxSum problem with input array $A$. Let $M_G$ be the max-weight $(d+1)$-clique in $G$. Then
			$$
				M_A=(d-1)M_G-(2^d-(d+1))c.
			$$
		\end{theorem}
		\begin{proof}
			Remember the definition of the \CMaxSum problem. We want to maximize the sum
			$$
				\sum_{j \in B_d}A[-i+\Delta\cdot j]
			$$
			over all choices of $d$-tuple $i$ and integer $\Delta$. We have an additional constraint that as we range over all $j \in B_d$, $type(-i+\Delta\cdot j)$ should range over all elements in $B_d$. We notice that $A[i]=-c$ if there are two $r \neq t \in [d]$ with $i_r,i_t\geq 0$. This means that the quantity we are maximizing
			\begin{align*}
				& \sum_{j \in B_d}A[-i+\Delta\cdot j]\\
				=& A[-i]+\left(\sum_{t \in [d]}A[-i+\Delta\cdot j^t]\right)-(2^d-(d+1))c.
			\end{align*}
			
			To prove the theorem, it suffices to show
			$$
				(d-1)M_G=M_A':=\max_{\substack{i \in [n]^d,\ \Delta \in [2n]\\ \text{s.t. }\Delta-i_1,\ldots,\Delta-i_d\ge 0}}A[-i]+\sum_{t \in [d]}A[-i+\Delta\cdot j^t].
			$$
			The equality follows from the following two cases.
			
			\paragraph*{Case $(d-1)M_G\geq M_A'$} If $A[-i]=-c$ or $A[-i+\Delta\cdot j^t]=-c$ (for some $t$), then we immediately get the inequality, by definitions of $A$, $D$ and $c$. We therefore assume that $A[-i]$ and each $A[-i+\Delta\cdot j^t]$ (for every $t \in [d]$) is equal to $D[i']$ for some $d$-tuple $i'$. Moreover, each one of these $d+1$ integers $D[i']$ is equal to the total edge-weight of $d$-clique induced by vertices $i'_1,\ldots,i'_d$ in $G$, since, otherwise, $D[i']=-c$ (see the definition of array $D$). By the construction, we have equality that $A[-i]=D[i]$. Fix $t$, and consider $A[-i+\Delta\cdot j^t]$. $A[-i+\Delta\cdot j^t]=D[i(t)]$ for some $d$-tuple $i(t)$. By equation \eqref{defa}, we must have
			$$
				-(i(t)-i(t)_t\cdot j^t)+j^t\cdot \|i(t)\|_1=-i+\Delta\cdot j^t,
			$$
			which, after simplification, yields
			$$
				i(t)_r=
				\begin{cases}
					i_r & \text{if }r \neq t\\
					\Delta-\|i\|_1 & \text{if }r=t.
				\end{cases}
			$$
			We conclude that $M_A'=D[i]+D[i(1)]+\ldots+D[i(d)]$, where $d$-tuple $i(t)$ is the same as $d$-tuple $i$, except that we replace entry $i_t$ by $\Delta-\|i\|_1$. Alternatively, $M_A'$ is the total edge-weight of $d$-cliques induced by sets of vertices $i,i(1),\ldots,i(d)$, which is the same as the total edge-weight of the $d+1$ clique induced by vertices $i_1,\ldots,i_d,\Delta-\|i\|_1$, multiplied by $d-1$. This yields the inequality. 
			
			\paragraph*{Case $(d-1)M_G\le M_A'$} Suppose that $M_G$ is achieved by $(d+1)$-clique induced by vertices $i_1,\ldots,i_d,i_{d+1}$. We set $i$ to be $d$-tuple $i=(i_1,\ldots,i_d)$ and we set integer $\Delta$ to be $\Delta=i_{d+1}+\|i\|_1$.
			Now we can check that $A[-i]+\sum_{t \in [d]}A[-i+\Delta\cdot j^t]$ is equal to the total edge-weight of $(d+1)$-tuple induced by vertices $i_1,\ldots,i_{d+1}$, multiplied by $d-1$. This statement follows from the definitions of arrays $A$ and $D$.
		\end{proof}

	\subsection{\CMaxSum $\Rightarrow$ \CMC}
		Let $A$ be the input $d$-dimensional array with side length $2n+1$ for the \CMaxSum problem. We produce $d$-dimensional array $A'$ of side length $2n+1$ from $A$ as follows. For all $d$-tuples $i \in \{-n,\ldots,n\}^d$, we set
		$$
			A'[i]=(-1)^{\|type(i)\|_1}A[i].
		$$
		$A'$ is input of \CMC problem.
		The correctness of this reduction follows from the definitions of both computational problems.
		
	\subsection{\CMC $\Rightarrow$ \MComb}
		Let $A$ be the input $d$-dimensional array with side length $2n+1$ for the \CMC problem. Let $M_C$ be the output of the \CMC problem on $A$. Let $c'$ be the largest absolute value among entries in $A$. We define $c:=100\cdot 2^d \cdot c'$. We define array $A'$ as follows.
		\begin{enumerate}
			\item Set $A'=A$.
			\item For every $i \in [n]^d$, set $A'[-i]=A[-i]+c$.
		\end{enumerate}
		Now we will show the equality
		\begin{equation} \label{cmcmc}
			M_C+c=
				\max_{i \in \{-n,\ldots,n\}^d,\ \Delta \in [2n+1]}\sum_{j \in B_d}(-1)^{\|j\|_1}\cdot A'[i+\Delta\cdot j].
		\end{equation}
		Notice that the r.h.s. of \eqref{cmcmc} is the \MComb problem on $d$-dimensional array with side length $2n+1$ after renumbering the entries. To show reduction, it therefore suffices to show equality \eqref{cmcmc}. Consider the $d$-tuple $i$ and the integer $\Delta$ that achieve the maximum in \eqref{cmcmc}. Suppose that for some $t \in [d]$, the $d$-tuple $i$ and integer $\Delta$ are such that $i_t+\Delta<0$. Then we have
		$$
			\sum_{j \in B_d}(-1)^{\|j\|_1}\cdot A'[i+\Delta\cdot j]\leq 2^d\cdot c'.
		$$
		because among the selected cells, all those with value $c$ cancel each other out.
		This cannot be an optimal solution, however, because we can achieve the value of at least $c-2^d\cdot c'>2^d\cdot c'$ by choosing $i=(-1,\ldots,-1)$ and $\Delta=1$. Therefore, an optimal choice of $d$-tuple $i$ and integer $\Delta$ will satisfy $i_t+\Delta\ge 0$ for all $t \in [d]$. If we add these constraints to the optimization problem on the r.h.s. of \eqref{cmcmc}, we get the \CMC problem with input array $A'$. The equality follows from the definition of array $A'$.
	
	\subsection{\MComb $\Rightarrow$ \MCube}
		Let $A$ be the input $d$-dimensional array with side-length $n$ to the \MComb problem. Given $A$, we produce $d$-dimensional array $A'$ of side-length $n-1$ such that the output of the \MCube problem on $A'$ is equal to the output of the \MComb problem on $A$.
		We construct $A'$ as follows. For every $d$-tuple $i \in [n-1]^d$, we set
		$$
			A'[i]=\sum_{j \in B_d}(-1)^{\|j\|_1}\cdot A[i+j].
		$$
		We can check equality
		\begin{equation} \label{mcmc}
			 \sum_{i_1\le k_1\le i_1+\Delta}\ldots \sum_{i_d\le k_d\le i_d+\Delta}A'[k_1,\ldots,k_d] 
			=  \sum_{j \in B_d}(-1)^{\|j\|_1}\cdot A[i+(\Delta+1)\cdot j].
		\end{equation}
		
		In the \MCube problem, we maximize l.h.s. of \eqref{mcmc} over $d$-tuples $i$ and integers $\Delta=\{0,\ldots,n-2\}$. In the \MComb problem, we maximize r.h.s. of \eqref{mcmc} over $d$-tuples $i$ and integers $(\Delta+1) \in [n]$. The reduction follows from the definitions of the computational problems.

\section{Hardness for \WDepth problem}\label{sec_wdepth}
In this section, we prove a matching lower bound for the \WDepth problem. 
We need to show that a $O(N^{(d/2)-\eps})$ algorithm for the \WDepth problem implies a $O(n^{d-2\eps})$ time algorithm for finding maximum-weight $d$-clique in an edge-weighted graph with $n$ vertices.
	
	For this purpose, we adapt a reduction from \cite{chan2008slightly}, where a conditional lower bound is shown for \emph{combinatorial} algorithms for the closely related Klee's measure problem.

\begin{theorem}
For any constant $\eps>0$, an $O\left(n^{\lfloor d/2\rfloor-\eps}\right)$ time algorithm for the \WDepth problem in $d$ dimensional space implies an $O\left(n^{d-2\eps}\right)$ time algorithm for the \wClique{d} problem.
\end{theorem}

	\begin{proof}
		For each $u \neq v \in V=[n]$ and $i \neq j \in [d]$, we create a rectangle 
		$$
			\left\{ (x_1,\ldots,x_d) \in [0,n)^d \ : \ x_i \in [u,u+1), x_j \in [v,v+1) \right\}
		$$
		and we set the weight of this rectangle to be equal to the weight $w(u,v)$ of the edge $(u,v)$.
		The total number of rectangles is $N=O(d^2n^2) = O(n^2)$.
		
		W.l.o.g., for all $u \neq v \in V$, $w(u,v)>0$ (if this is not so, we add a sufficiently large enough fixed quantity to the weight of every edge). The heaviest point $p$ therefore lives in $[0,n)^d$.
		We claim that the weight of the heaviest point in $[0,n)^d$ is twice the weight of the heaviest $d$-clique in the graph.
		This is so, since the weight of a point $p \in [0,n)^d$ is equal to 
		$$
			\sum_{i \neq j \in [d]}w(\lfloor p_i \rfloor, \lfloor p_j \rfloor),
		$$
		which is twice the weight of $d$-clique supported on the vertices $\lfloor p_1 \rfloor,\ldots,\lfloor p_d \rfloor$. Conversely, the weight of $d$-clique supported on the vertices $v_1,\ldots,v_d \in [n]$, is equal to half of the weight of point $(v_1,\ldots,v_d) \in [0,n)^d$.
	\end{proof}

\section{Acknowledgments}
We thank Piotr Indyk for providing helpful discussions and comments on an earlier version of the paper.
We thank Linda Lynch for copy-editing parts of the paper. One of the authors, Nishanth Dikkala, would like to thank Namratha Dikkala for a discussion about efficient algorithms for geometric problems which led to some of the questions pursued in this paper.

\bibliographystyle{alpha}
\bibliography{ref}

\end{document}